\documentclass[letterpaper, 10 pt, conference]{ieeeconf}  

\IEEEoverridecommandlockouts                              

\title{\LARGE \bf
Adaptive Rejection of Periodic Disturbances Acting on Linear Systems with Unknown Dynamics
}
\author{Behrooz Shahsavari, Jinwen Pan and Roberto Horowitz
\thanks{B. Shahsavari and R. Horowitz are with the Department of Mechanical Engineering, University of California, Berkeley, CA 94720, U.S.A.
        {\tt\small \{behrooz,horowitz\}@berkeley.edu}}%
\thanks{J. Pan is with the Department of Automation, University of Science and Technology of China, Hefei, 230027, P.R.China,
        {\tt\small jinwen@berkeley.edu}}%
}
\usepackage{graphicx}          
\usepackage{amsmath}
\usepackage{color}
\usepackage{arydshln}
\usepackage{mathrsfs}
\usepackage{amsbsy}
\usepackage{ntheorem}
\usepackage{float}
\usepackage{graphicx}
\usepackage[T1]{fontenc} 
\usepackage[colorlinks,urlcolor=blue,citecolor=blue,linkcolor=blue,pdfusetitle]{hyperref}
\usepackage{pdflscape} 
\usepackage{amsmath,amssymb}
\usepackage{longtable}
\usepackage{algorithm}
\usepackage{algorithmicx}
\usepackage{algpseudocode}
\usepackage[usenames,dvipsnames]{xcolor}
\usepackage{subcaption}
\usepackage{widetext}

\theoremstyle{break}
\newtheorem{theorem}{Theorem}
\newtheorem{lemma}{Lemma}

\newtheorem{remark}{Remark}

\newcommand{\qi}{{q^{-1}}}
\newcommand{\tha}{ {\hat{\theta}} }

\newcommand{\bm}{\left[\begin{array} }
\newcommand{\mb}{\end{array}\right] }
\newcommand{\bk}{\left\{\begin{array} }
\newcommand{\kb}{\end{array}\right\} }

\newcommand{\figpath}{.}

\begin{document}

\maketitle
\thispagestyle{plain}
\pagestyle{plain}
%
%

\begin{abstract}  
This paper proposes a novel direct adaptive control method for rejecting unknown deterministic disturbances and tracking unknown trajectories in systems with uncertain dynamics when the disturbances or trajectories are the summation of multiple sinusoids with known frequencies, such as periodic profiles or disturbances. The proposed algorithm does not require a model of the plant dynamics and does not use batches of measurements in the adaptation process. Moreover, it is applicable to both minimum and non--minimum phase plants. The algorithm is a ``direct'' adaptive method, in the sense that the identification of system parameters and the control design are performed simultaneously. In order to verify the effectiveness of the proposed method, an \emph{add--on} controller is designed and implemented in the servo system of a hard disk drive to track unknown nano--scale periodic trajectories.
\end{abstract}

\section{Introduction}
Control methodologies for rejecting periodic and multi--harmonic disturbances or tracking such trajectories 
have attracted many researchers in the past two decades. 
There is a multitude of applications, especially due to the dominating role of rotary actuators and power generators, that crucially rely on this type of control. 
A non--exhaustive list of these applications include aircraft interior noise control \cite{emborg1998cabin,wilby1996aircraft}, 
periodic load compensation in wind turbines \cite{stol2003periodic,houtzager2013rejection}, 
wafer stage platform control \cite{de2000synthesis,dijkstra2004iterative}, 
steel casting processes \cite{tsao1996rejection}, 
laser systems \cite{mcever2004adaptive}, milling machines \cite{rober1996control,tsao1991control} and hard disk drives \cite{shahsavari2015adaptive,shahsavari2014adaptive}.
In this paper, we introduce a novel direct adaptive control for rejecting deterministic disturbances and tracking unknown trajectories in systems with unknown dynamics when the disturbances or trajectories are the summation of multiple sinusoids with known frequencies. Note that a periodic disturbance/trajectory with a known period can be considered as a special case of the problems under our study.

Control methods applied to this class of problems are typically categorized into two types, namely feedback methods that are based on internal model principle (IMP)  \cite{francis1976internal} and feedforward algorithms that usually use an external model \cite{tomizuka1990disturbance} or a \emph{reference} signal correlated with the disturbance \cite{bodson1997adaptive}.
The classical form of internal model for periodic disturbances introduces poles on the stability boundary which can cause poor numerical properties and instability when implemented on an embedded system with finite precision arithmetic. Instability can also happen due to unmodeled dynamics when the poles are on or very close to the stability boundary \cite{bodson2005rejection}.
Another limitation of IMP based repetitive control is that the controller sampling frequency has to be divisible by the fundamental frequency of the disturbance.

In general, adaptive feedforward control for this class of problems does not have the above limitations. This type of controllers commonly estimate a set of parameters that determine the control law. The estimated parameters converge  when the system is not stochastic or the adaptation gain is vanishing in a stationary (or cyclostationary) stochastic system. This implies that the estimated parameters can be frozen after convergence and the control sequence becomes a pure feedforward action that can be stored and then looked up without a need to feeding the error to the controller \cite{shahsavari2015adaptive}. This is an important advantage over the feedback schemes because that type of controller should be constantly in the loop to generate the control sequence.
Another advantage is that the Bode's sensitivity integral theorem does not hold true, which implies that perfect rejection can be achieved without affecting the suppression level at other frequencies.
Nevertheless, analysis of the adaptive methods is, in general, more complex and relies on a set of assumptions that may not hold true in many situations.

Although rejection of sinusoidal disturbances is a classical control problem, few algorithms exist for the case that the system dynamics is unknown and possibly time--varying. Gradient descent based algorithms that use online identification schemes to obtain a finite impulse response (FIR) of the plant \cite{zhang2001cross}  have been proposed 
by both control and signal processing communities.
The number of estimated parameters in these methods is usually large since low--order FIR models cannot mimic complex dynamics. 
The harmonic steady--state (HSS) control is another adaptive method for rejection of sinusoidal disturbances
which is easy to understand and implement.
However, it suffers from slow convergence since it relies on averaging over batches of data \cite{patt2005higher,chandrasekar2006adaptive}.

This paper provides a novel direct adaptive feedforward control that does not require a model of the plant dynamics or batches of measurements. The proposed method does not rely on any assumptions about the location of plant zeros and can be applied to minimum and non--minimum phase plants. 
The algorithm is a ``direct'' adaptive method, meaning that the identification of system parameters and the control design are not separate. 
It will be shown that the number of adapted parameters in this scheme is significantly less than methods that identify the plant frequency response  when the number of frequencies is large \cite{chandrasekar2006adaptive,pigg2006adaptive,pigg2010adaptive}.

The remainder of this paper is organized as follows.
We first formalize the problem and explain the system under our study in section \ref{sec:arch}.
Mathematical preliminaries and notations are given in section~\ref{sec:mathandprelim}.
The algorithm derivation is presented in section \ref{sec:adaptivecontrolsynthesis} and simulation results for 
rejecting periodic disturbances in a hard disk drive nanopositioner servo system are illustrated in section \ref{sec:dir_results}.
Conclusion remarks and future work form section \ref{sec:conclusion}.

\section{Problem Statement}\label{sec:arch}

\begin{figure}[t]
	\centering
	\centering
\includegraphics[width=0.7\linewidth]{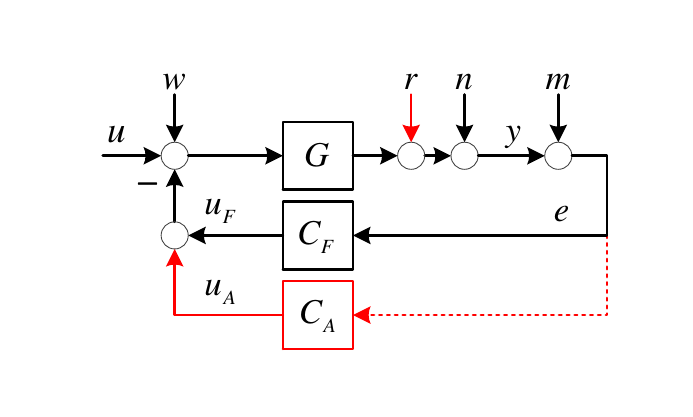}
	%
	\caption{Closed loop system augmented by a \emph{plug--in} adaptive controller.}
	\label{fig:SimpleFeedBackAndAdaptive_Abstract}
\end{figure}
The adaptive controller proposed in this work is aimed to be implemented in a \emph{plug--in} fashion,
meaning that it is used to augment an existing robustly stable closed--loop system in order to reject periodic disturbances (track periodic trajectories) that are not well rejected (tracked) by the existing controller. 
In this architecture, the original controller can be designed  without consideration of this special control task. Moreover, the plug--in controller does not alter the performance of the original control system.
To clarify this notion, we use a common Single-Input Single-Output (SISO) plant--controller interconnection shown in Fig.~\ref{fig:SimpleFeedBackAndAdaptive_Abstract} as a running example. 
The blocks $G$ and $C_F$ in the figure respectively denote a linear time invariant (LTI) plant and an LTI feedback compensator that form a stable closed--loop sensitivity function
	\begin{align*}
S := \frac{1}{1+G C_F}.
	\end{align*} 
One of the main contributions of the controller that will be presented shortly is that it does not require the plant and controller dynamics.
Since our design does not depend on whether the plant/nominal controller are continuous or discrete time, we assume that both $G$ and $C_F$ are discrete time systems to make notations simpler. 
It is worth noting that this interconnection is only a running example through this paper and the proposed controller can be \textit{plugged} to any unknown and stable LTI system regardless of its internal stabilization mechanism.  

A general stochastic environment is considered for the system by appending input disturbance $w$, output disturbance $n$, and contaminating measurement noise $m$ to our framework. Generally, the nominal feedback controller $C_F$ is designed to compensate for these input and output noises. 
The special deterministic disturbance that should be compensated by the adaptive controller is denoted by $r$, and without loss of generality, we assume that it contaminates the plant output. Let the adaptive controller sampling time be $T$. The class of disturbances under our study can be written as
\begin{align}\label{eq:r_definition}
r(k)=\sum\limits_{i=1}^{n} \alpha_i \sin\left(\omega_i k T + \delta_i \right)
\end{align}
where the amplitudes, $\alpha_i$, and phase shifts, $\delta_i$, are unknown but the frequencies, $\omega_i$, are known.

Our objective is to synthesize an adaptive controller that only uses the scalar--valued error signal $e(k)$ to
generate a feedforward control $u_{A}(k)$ such that it perfectly compensates for the effect of $r(k)$ on the error signal $e(k)$.
We call it a feedforward controller because when the system dynamics and disturbance profile are time--invariant, $C_A$ will not depend on the error signal once the control law is learned. 
\section{Mathematical Preliminaries and Notations}\label{sec:mathandprelim}
Let $R(z)$ be the single--input single--output system (in $z$--domain) from the adaptive control injection points to the error signal $e$
	\begin{align*}
	R\left(z\right) :&= G\left(z\right) S\left(z\right)
	\end{align*}
and let $A(\qi)$ and $B(\qi)$ be the polynomials that represent $R(\cdot)$ in time--domain
	\begin{align}\label{eq:R_definedby_AB}
	\begin{split}
		R(\qi) &= \frac{B(\qi)}{A(\qi)} 
.
	\end{split}
	\end{align}
We define an output disturbance, say $\bar{w}$, on $R(\cdot)$ and a polynomial $C(\qi)$ that satisfies
	\begin{align*}
		 \frac{1}{A(\qi)} \left[\bar{w}(k)\right]&= R\left[w(k)\right]+S\left[n(k) + m(k)\right] \\
	\end{align*}
We also define $\bar{r}(k)$ as a sequence on the output of $R(.)$ that has the same effect as $r$ in the error signal $e$
\begin{align*}
\bar{r}(k) = S\left[r(k)\right]\cdot
\end{align*} 
When the same transfer function filters multiple input signals $i^1(k), i^2(k), \cdots, i^m(k)$, we abuse the notation and use 
\begin{align*}
\begin{bmatrix}  T[i^1(k)]\\ T[i^2(k)]\\ \vdots \\T[i^m(k)]
\end{bmatrix}
=
T\left[\begin{bmatrix}i^1(k)\\ {i^2(k)}\\ \vdots\\ {i^m(k)}\end{bmatrix}\right].
\end{align*}
%
The disturbance $r(k)$ in \eqref{eq:r_definition} can be factorized as the inner product of a known ``regressor'' vector $\phi(k)$ and an unknown
vector of parameters $\theta$ 
	\begin{align}\label{eq:r_thetaphi}
	 r(k)= \theta^T \phi_R(k)
	\end{align}
where 
	\begin{align*}
	\theta^T:&=\left[\alpha_1 \cos(\delta_2), \alpha_1 \sin(\delta_2),\ldots \right.\\
	& \left. \qquad \qquad, \alpha_n \cos(\delta_n),\alpha_n \sin(\delta_{n})\right]\\
	\phi_R^T(k):&=\left[\sin(\omega_1kT), \cos(\omega_1kT),\ldots \right.\\
	\notag
	&\left. \qquad \qquad ,\sin(\omega_nkT), \cos(\omega_nkT)  \right].
	\end{align*}
\begin{lemma}\label{lem:periodic-response}
	Consider $r(k)$ as a general periodic signal and $L(\qi)$ as a discrete-time linear system.  
	The steady state response ${\tilde{r}}(k):=L(\qi)[{r}(k)]$ is periodic. 
	Moreover, when $r(k)$ is a linear combination of sinusoidal signals factorized similar to \eqref{eq:r_thetaphi}, ${\tilde{r}}(k)$ (in steady state) consists of sinusoidals with the same frequencies 
		\begin{align*}
			{\tilde{r}}(k)&=L(\qi)\left[\theta^T \phi_{{R}}(k)\right] 
			= \theta^T L(\qi)\left[ \phi_{{R}}(k)\right]
			= \theta^T \phi_{R_L}(k)
		\end{align*}
	where
		\begin{gather}
			\notag
			\phi_{R_L}^T(k):=
			\left[{m_{L_1}}\sin(\omega_1kT+\delta_{L_1}), {m_{L_1}}\cos(\omega_1kT+\delta_{L_1}),  \right. \\
			\qquad  \left. \ldots, {m_{L_n}}\sin(\omega_nkT+\delta_{L_n}), {m_{L_n}}\cos(\omega_nkT+\delta_{L_n})  \right].
		\end{gather} 
	Here, $m_{Li}$ and $\delta_{Li}$ are the magnitude and phase of $L(e^{-j \omega_iT})$ respectively. 
	Define 
		\begin{align}\label{eq:Dai}
			\begin{split}
			D_{Li}:=\begin{bmatrix}
				m_{Li} \cos(\delta_{Li} ) & m_{Li} \sin(\delta_{Li} )\\
				-m_{Li} \sin(\delta_{Li} ) & m_{Li} \cos(\delta_{Li} )
				\end{bmatrix}
			\end{split},
		\end{align}
	$\phi_{{R}}(k)$ can be transformed to $\phi_{R_L}(k)$ by a linear transformation
	\begin{align}\label{eq:Da}
	\phi_{R_L}(k)=\underbrace{\begin{bmatrix}
	D_{L1} & 0 & \cdots & 0\\
	0 & D_{L2} & \cdots & 0\\
	\vdots & \vdots &\ddots & \vdots\\
	0 & 0 & \cdots & D_{Ln}
	\end{bmatrix}}_{D_L} \phi_R(k).
	\end{align}
	As a result 
	\begin{align*}
	\tilde{r}(k)=L(\qi)\left[\theta^T \phi_{{R}}(k)\right] = \theta^T \phi_{R_L}(k) =\theta_{\bar{R}}^T D_L \phi_R.
	\end{align*}
\end{lemma}
	\begin{proof}
	Refer to \cite{kamen2000fundamentals} for a general formula for 
	the steady-state sinusoidal response of a linear time-invariant system.
	\end{proof}

Using Lemma \ref{lem:periodic-response} the equivalent disturbance $\bar{r}$ can be factorized as
	\begin{align}\label{eq:direct_rbar}
		\bar{r}(k)&=\theta_{\bar{R}}^T \phi_{{R}}(k)
	\end{align}
where $\theta_{\bar{R}}$ is unknown.

\section{Adaptive Control Synthesis}\label{sec:adaptivecontrolsynthesis}

\subsection{Error Dynamics}
The error sequence in time domain can be represented as a function of the closed--loop dynamics, control signals and disturbances
	\begin{align}\label{eq:mainE}
	{e}(k) &= \frac{B(\qi)}{A(\qi)}\left(u(k)+ u_{A}(k)\right) + \frac{1}{A(\qi)} \bar{w}(k) + \bar{r}(k)
	\end{align}
where $u(k)$ is an exogenous excitation signal and  $\bar{w}(k)$ is an unmeasurable wide--sense stationary sequence of independent random values with finite moments.
We assume that the nominal feedback controller is able to stabilize the open loop plant, i.e. $A(p)$ has all roots strictly outside the unit circle. 
Although a real dynamic system cannot be exactly described by finite order polynomials, in most of applications $A$ and, $B$  can be determined such that they give a finite vector difference equation  that describes the recorded data as well as possible, i.e.
\begin{align*}
A(\qi):&= 1+a_1 \qi + a_2 q^{-2} + \cdots + a_{n_A}  q^{-n_A} \\
B(\qi):&= 
		b_1 q^{-1} + b_2 q^{-2} + \cdots + b_{n_{B}}  q^{-n_{B}} 
.
\end{align*}
Without loss of generality, we assume that the relative degree of the transfer function from the controllable input channel to the output is 1, which implies that $n_A=n_{B}$. The analysis for other non-negative relative degrees is very similar to the sequel, but the notation would be more tedious due to differences in vector/matrix sizes.
Let $A^*(\qi):= 1 - A(\qi)$, then the error is given by 
\begin{align}
\begin{split}
\label{out1} {e}(k) &= A^*(\qi) {e}(k)+B(\qi) \left(u(k)+ u_{A}(k)\right)\\
& \qquad \qquad \qquad +  \bar{w}(k) 
+ A(\qi)\bar{r}(k)
\end{split}
\end{align}
This equation can be represented purely in discrete time domain in a vector form
\begin{align}
\begin{split}
\label{out2} {e}(k) &= \theta_A^T \phi_{e}(k)+ \theta_B^T \left(\phi_u(k)+\phi_{ u_{A}}(k)\right) 
 + {\tilde{r}}(k) + {\bar{w}}(k)
\end{split}
\end{align}
where 
\begin{align}\label{eq:direct_theta_phi}
	\begin{split}
	\theta_A^T :&= \begin{bmatrix} -a_1, &-a_2, &\cdots, &-a_{n_A}   \end{bmatrix},\\
	\theta_B^T :&= \begin{bmatrix} b_1, & b_2,  & \cdots, & b_{n_A}   \end{bmatrix},\\
	\phi_e^T(k) :&= \begin{bmatrix} {e}(k-1), & {e}(k-2), & \cdots, & {e}(k-n_A)   \end{bmatrix},\\
	\phi_u^T(k) :&= \begin{bmatrix} u(k-1), & u(k-2),  & \cdots, & u(k-n_A)   \end{bmatrix},\\
	\phi_{u_{A}}^T(k) :&= \begin{bmatrix}  u_{A}(k-1), &  u_{A}(k-2),  & \cdots, &  u_{A}(k-n_A)   \end{bmatrix},\\
	\end{split}
\end{align}
and ${\tilde{r}}(k):=A(\qi)\bar{r}(k)$. 
Note that two regressors, denoted by $\phi_u(k)$ and $\phi_{ u_{A}}(k)$, are considered for the excitation signal $u(k)$ and the adaptive control $u_A(k)$ separately although they could be combined into a unique regressor.  The rationale behind this consideration will be explained later after \eqref{eq:e}.
Since  disturbance $\bar{r}(k)$ is periodic and $A(\qi)$ is a stable filter -- i.e. it operates as an FIR filter -- the response ${\tilde{r}}(k)$ is also periodic by Lemma~\ref{lem:periodic-response}
	\begin{align*}
	{\tilde{r}}(k)&=A(\qi)\left[\theta_{\bar{R}}^T \phi_{{R}}(k)\right] 
	= \theta_{\bar{R}}^T \phi_{R_A}(k)
	\end{align*}
where
	\begin{align}
		\phi_{R_A}(k)=\underbrace{\begin{bmatrix}
		D_{A1} & 0 & \cdots & 0\\
		0 & D_{A2} & \cdots & 0\\
		\vdots & \vdots &\ddots & \vdots\\
		0 & 0 & \cdots & D_{An}
		\end{bmatrix}}_{D_A} \phi_R(k).
	\end{align}
Accordingly, ${\tilde{r}}(k)$ can be represented using the same regressor vector, $\phi_R(k)$
	\begin{align*}
		{\tilde{r}}(k)&= \theta_{\bar{R}}^T \phi_{R_A}(k) 
		= \underbrace{\theta_{\bar{R}}^T D_A}_{ \theta_{{R}}^T} \phi_{{R}}(k) 
		= \theta_{{R}}^T \phi_{{R}}(k).
	\end{align*}
Substituting this expression in \eqref{out2} yields
	\begin{align}\label{out3} 
	\begin{split}
		{e}(k) &= \theta_A^T \phi_{e}(k)+ \theta_B^T \left(\phi_u(k)+\phi_{ u_{A}}(k)\right)+ \theta_R^T \phi_{{R}}(k)+ {\bar{w}}(k).
	\end{split}
	\end{align}
Equation \eqref{out1} shows that an \emph\textit{ideal} control signal $u_A^*(k)$ should satisfy
	\begin{align}\label{eq:bdar}
		B(\qi) u_A^*(k) + A(\qi) \bar{r}(k)=0.
	\end{align}
Again, since $B(\qi)$ and $A(\qi)$ are both LTI systems and $\bar{r}(k)$ contains only sinusoidal signals, the ideal control signal ${u_A^*}(k)$ has to have sinusoidal contents at frequencies equal to $\omega_i$'s. This motivates us to decompose the ideal control signal into
	\begin{align*}
	{u_A^*}(k)=\theta_D^T\phi_R(k).
	\end{align*}
By this representation of the control signal, our goal will be to estimate $\theta_D$ in an adaptive manner. We define the \emph{actual} control signal as 
	\begin{align}\label{eq:control}
	 u_{A}(k)=\tha_D^T(k)\phi_R(k)
	\end{align}
where $\tha_D(k)$ is the vector of estimated parameters that should ideally converge to $\theta_D$. As a result, the residual in \eqref{eq:bdar} when $\theta_D$  is replaced by $\hat{\theta}_D$ is 
	\begin{align}\label{eq:swap1}
		\begin{split}
			&B(\qi)  u_{A}(k)+A(\qi)\bar{r}(k) = \\
			& \qquad \qquad \qquad B(\qi)\left[\tha_D^T(k)\phi_R(k)\right]+\theta_R^T \phi_R(k).
		\end{split}
	\end{align}
\begin{lemma}\label{lemma:swapping}
	Let $B(\qi)$ have a minimal realization $B(\qi)=C_B^T\left(qI-A_B\right)^{-1}B_B$. Then, 
	\begin{align*}
	B(\qi)\left[\tha_D^T(k)\phi_R(k)\right] = \tha_D^T(k) \left\{B(\qi)\left[\phi_R(k)\right]\right\}+w_t(k)
	\end{align*}
	where 
	\begin{align*}
	w_t(k):=& -M(\qi)\left[H(\qi)\left[\phi_R^T(k+1)\right]
	\left[\tha(k+1)-\tha(k)\right]\right]\\
	M(\qi):&=C_B^T(qI-A)^{-1}\\
	H(\qi):&=(qI-A)^{-1}B_B
	\end{align*}
\end{lemma}
	\begin{proof}
		Refer to the discrete--time swapping lemma in \cite{tao2003adaptive}.
	\end{proof}
We define a new parameter vector
\begin{align}\label{eq:thetaM}
\theta_M^T(k):=\tha_D^T(k) D_B +\theta_R^T
\end{align}
where $D_B$ is a  matrix similar to $D_A$ in \eqref{eq:Da}, but its block diagonal terms are formed by the magnitude and phase of $B\left(e^{-j\omega_iT}\right)$ rather than $A(e^{-j\omega_iT})$.
Since the vector $\theta_M(k)$ corresponds to the imperfection in control synthesis, it is called the \emph{residual parameters vector} throughout this section.
Substituting the result of Lemma~\ref{lemma:swapping} in \eqref{eq:swap1} yields
	\begin{align}\notag
		B(\qi)  &u_{A}(k)+ A(\qi)\bar{r}(k) \\
		\notag
		&=\tha_D^T(k)\left(B(\qi)[\phi_R(k)]\right)+\theta_R^T \phi_R(k)+w_t(k)\\ 
		\notag
		&=\tha_D^T(k) D_B \phi_R(k)+\theta_R^T \phi_R(k)+w_t(k)\\
		\label{eq:bu+ar}
		&=\theta_M^T(k) \phi_R(k) + w_t(k)
	\end{align}
where the term $\theta_M^T(k) \phi_R(k)$ corresponds to the residual error at the compensation frequencies. The term $w_t(k)$ represents the transient excitation caused by the variation of $\tha_D(k)$ over time.
As a result, the term 
$\theta_B^T \phi_{ u_{A}}(k) + \theta_R^T \phi_{{R}}(k)$
in \eqref{out3} can be replaced by \eqref{eq:bu+ar}  which yields to
\begin{align}\label{eq:e}
\begin{split}
&{e}(k) = \theta_A^T \phi_{e}(k)+ \theta_B^T \phi_u(k) + \theta_M^T(k) \phi_{{R}}(k)+ w_t(k)+{\bar{w}}(k).
\end{split}
\end{align}
\begin{remark}
The reason behind choosing two separate regressors for $u(k)$ and $u_A(k)$, as remarked earlier, is that the recent substitution in the above equation is not feasible if the two regressors were combined into a single regressor.
\end{remark}

\subsection{Parameter Adaptation Algorithm}
The error dynamics shows that the information obtained from measurements cannot be directly used to estimate $\hat{\theta}_D$ as long as the closed loop system is unknown. We propose an adaptive algorithm in this section that accomplishes the estimation of the closed loop system and control synthesis simultaneously. 

Let $\tha_A$, $\tha_B$ and $\tha_M$ be the estimated parameters analogous to \eqref{eq:e}. 
We denote the \emph{a--priori} estimate of the error signal at time $k$ based on the estimates at $k-1$ as 
	\begin{align}\label{eq:yapriori}
	\begin{split}
		\hat{y}(k)&=\tha_A^T(k-1) \phi_{e}(k)+ \tha_B^T(k-1) \phi_u(k) + \tha_M^T(k-1) \phi_{{R}}(k).
	\end{split}
	\end{align}
and accordingly, the \emph{a--priori} 
estimation error is defined as
	\begin{align}\label{eq:eapriori}
	\begin{split}
	\epsilon^{\circ}(k):&={e}(k) - \hat{y}(k)
	\end{split}
	\end{align}
%
Assume that the estimates at time $k=0$ are initialized by either zero or some ``good'' values when prior knowledge about the system dynamics  is available.  
We propose the following adaptation algorithm for updating the estimated parameters
	\begin{align}\label{eq:paa}
	\begin{split}
		&\bm{c} \tha_A(k) \\ \tha_B(k) \\ \hdashline \tha_M(k) \mb = 
		\bm{c} \tha_A(k-1) \\ \tha_B(k-1)  \\ \hdashline \tha_M(k-1) \mb \\
		&\quad +
		\bm{c:c} \gamma_1(k) F^{-1}(k) & 0 \\ \hdashline 0  & \gamma_2(k)f^{-1}(k)I \mb
		\bm{c} \phi_{e}(k) \\  \phi_u(k) \\  \hdashline \phi_R(k) \mb \epsilon^{\circ}(k).
	\end{split}
	\end{align}
$F(k)$ is a positive (semi) definite matrix with proper dimension and $f(k)$ is a positive scalar. 
These gains, which are usually known as learning factor or step size, can be updated via either recursive least squares algorithm, least mean squares type methods or a combination of them. 
We use recursive least squares for the plant  since the number of coefficients is usually ``small''. On the other hand, for large $n$, the recursive least squares algorithm requires major computations. Therefore, it is of interest to reduce the computations, possibly at the price of slower convergence, by replacing the recursive least squares update law by the stochastic gradient method.
It is well known that the step size of adaptive algorithms in stochastic environments should converge to zero or very small values to avoid ``excess error'' caused by parameter variations due to noises. Therefore, positive real valued decreasing scalar sequences $\gamma_1(k)$ and $\gamma_2(k)$ are considered conjointly with the step sizes. 
More explicitly, the update rules for $F$ and $f$ are
	\begin{gather*}
		F(k)=F(k-1)+\gamma_1(k)\left(\begin{bmatrix}\phi_{e}(k) \\ \phi_u(k)\\ \phi_{\epsilon}(k)  \end{bmatrix}\begin{bmatrix}\phi_{e}(k) \\ \phi_u(k)\\ \phi_{\epsilon}(k)  \end{bmatrix}^T-F(k-1)\right)\\
		f(k)=f(k-1)+\gamma_2(k)\left(\phi_R^T(k)\phi_R(k)-f(k-1)\right).
	\end{gather*}
\begin{remark}
$\phi_u(k)$ should be persistently exciting of order $2n_A$ in order to guarantee that $F(k)$ is non--singular and \eqref{eq:paa} is not susceptible to numerical problems. It is clear that  $f(k)$ is not subjected to this issue since $\phi_R^T(k)\phi_R(k)$ is always strictly positive.
\end{remark}



%
\subsection{Control Synthesis}
Suppose that the parameter vector $\theta_M(k)$ and response matrix $D_B$ are known at time step $k$. Then, a possible update rule that satisfies \eqref{eq:bdar} would be
	\begin{align}\label{eq:direct_ideal_solution}
	\begin{split}
		0 &= \tha_D^T(k+1)D_B+\theta_R^T \\
		&=\tha_D^T(k+1)D_B+\theta_M^T(k)-\tha_D^T(k)D_B\Rightarrow\\
		\tha_D^T(k+1) &= \tha_D^T(k) -  \theta_M^T(k) D_B^{-1}.
	\end{split}
	\end{align}
Here, we have used the fact that $D_B$ is a block diagonal combination of scaled rotation matrices, which implies that it is full rank and invertible. This is an infeasible update rule since neither $\theta_M(k)$, nor $D_B$ is known. We replace these variables by their respective estimated values and use a small step size $\alpha$ in order to avoid large transient and excess error
	\begin{align*}
		\tha_D^T(k+1) &=  \tha_D^T(k) - \alpha \tha_M^T(k) \hat{D}_B^{-1}(k).
	\end{align*}
Note that this update rule works as a first order system that has a pole at $1$.
In order to robustify this difference equation we alternatively propose using a \emph{Ridge} solution for \eqref{eq:direct_ideal_solution}. More formally, we are interested in minimizing the instantaneous cost function 
	\begin{align*}
		J_c(k) := \frac{1}{2}\|\tha_D^T(k)+\theta_R^T {D}_B^{-1}\|_2^2+\frac{1}{2}\lambda \|\tha_D^T(k)\|_2^2
	\end{align*}
where $\lambda$ is a (positive) weight for the penalization term.
We use a gradient descent algorithm to recursively update $\tha_D$. 
Let $\beta = 1-\alpha \lambda$ and 
the gradient of $J_c(k)$ with respect to $\tha_D^T(k)$ be denoted by
	\begin{align*}
		\frac{\partial J_c(k)}{\partial \tha_D^T(k)}=\left(\tha_D^T(k)+\theta_M^T(k){D}_B^{-1}-\tha_D^T(k)\right)+\lambda \tha_D^T(k).
	\end{align*}
Since the actual values of $\theta_M$ and $D_B$ are unknown, we use the estimates and define 
the gradient descent update rule for $\tha_D$ as
	\begin{align}		\label{eq:integral}
		\tha_D^T(k+1) 
		&=  \beta \tha_D^T(k) - \alpha \tha_M^T(k) \hat{D}_B^{-1}(k).
	\end{align}
This expression implies that a positive value of $\beta$  less than 1 results in a bounded value of $\hat{\theta}_D$ in steady state as long as $ \tha_M^T \hat{D}^{-1}_B$ stays bounded. Moreover, assuming that $\tha_M$ and $\hat{D}_B$ converge to the actual $\theta_M$ and $D_B$ -- which will be proved later -- the steady state residue is
	\begin{align*}
		\lim\limits_{k\rightarrow \infty}\theta_M(k) &=D_B^T\lim\limits_{k\rightarrow \infty }\hat{\theta}_D(k)+\theta_R\\ 
		&=\frac{-\alpha}{1-\beta } \lim\limits_{k\rightarrow \infty}\theta_M(k) +\theta_R\\
		&=\frac{1-\beta}{1-\beta + \alpha} \theta_R.
	\end{align*}
This expression shows that there is a compromise between the steady state attenuation level and robustness, and in order to achieve both, the two gains should be chosen such that
\begin{align}\label{eq:alpha_beta_relation}
0< \alpha \ll \beta < 1.
\end{align}
Now that we have an update law for $\tha_D(k)$, we have a complete algorithm for synthesizing the control signal (repeated from \eqref{eq:control})
	\begin{align*}
	u_{A}(k)=\tha_D^T(k)\phi_R(k).
	\end{align*}

\begin{theorem}\label{thm:main}
The control update rule outlined by \eqref{eq:control},  \eqref{eq:paa} and \eqref{eq:integral} make $\theta_M$ converge to $\frac{1-\beta}{1-\beta+\alpha} \theta_R$ with probability 1, 
	the only equilibrium point of the closed loop system is stable in the sense of Lyapunov
	 and it corresponds to $\tha_A = \theta_A$,
	$\tha_B = \theta_B$, and $\tha_M = \frac{1-\beta}{1-\beta+\alpha} \theta_R$
	\textbf{if} the following 
	conditions are satisfied:
	\begin{enumerate}
		\item \label{ass:thm:u}
			$u(k)$ is persistently exciting of at least order $2n_A$.
		\item \label{ass:thm:gamma}
			$\sum\limits_{k=1}^{\infty}\gamma_i(k)=\infty$ and $\gamma_i(k)\rightarrow 0$ as $k\rightarrow \infty$ for $i=1,2$.
		\item  \label{ass:thm:A}
			The estimated $\tha_A(k)$ belongs to 
			\begin{align*}
			\begin{split}
			\mathscr{D}_A:=&\left\{\vphantom{\frac14} \tha_A : 1+\hat{a}_1 q +\cdots + \hat{a}_{n_A} q^{{n_A}} =0 \Rightarrow |q|>1 \vphantom{\frac14}\right\}.
			\end{split}
			\end{align*}
			infinitely often with probability one.
		\item \label{ass:thm:B}
			The estimated $\tha_B(k)$ always belongs to
			\begin{align*}
				\mathscr{D}_B:&=
				\left\{\vphantom{\frac12}  \tha_B : 0< |\hat{b}_1 e^{-j\omega_h} + \cdots + \hat{b}_{n_A}  e^{-jn_A \omega_h}|\right. \\
				&\qquad~~ < \frac{\alpha}{1-\beta}|b_1 e^{-j\omega_h} + \cdots + b_{n_A}  e^{-jn_A \omega_h}|,\\
				&\qquad \left.  \text{  and  } \forall h \in \{1,\ldots,n\} \vphantom{\frac12}  \right\}.		\end{align*}
		\item   \label{ass:thm:H}
			$\text{Real}\left(\frac{{B}(e^{-j\omega_h})}{\hat{B}_k(e^{-j\omega_h})}\right)>0$  for all $h\in \{1, \ldots, n\}$ infinitely often with probability one.
		\newcounter{assumptionCounter}
		\setcounter{assumptionCounter}{\theenumi}
	\end{enumerate}
\end{theorem}
\begin{proof}
Only the sketch of the proof is outlined here due to space limitation.	The method of analysis of stochastic recursive algorithms developed by Ljung \cite{ljung1977analysis} can be deployed to study the
	convergence and asymptotic behavior of the proposed adaptive algorithm with update and control rules given in \eqref{eq:paa}, \eqref{eq:integral} and
	\eqref{eq:control}. 
	Three sets of regularity conditions are proposed in \cite{ljung1977analysis} that
	target the analysis of deterministic and stationary stochastic processes.
	The problem under our study cannot be exactly outlined in these frameworks since the input signal consists of stochastic and deterministic parts,
	and as a matter of fact, it is a cyclostationary stochastic process.
	However, ``Assumptions C'' in \cite{ljung1977analysis} can be adopted and generalized to this case with minor modifications.
	It can be shown that these regularity conditions are satisfied when the assumptions of theorem \ref{thm:main} are satisfied. 
	Under these regularity conditions, the results of theorem 2 in \cite{ljung1977analysis}  imply that the only convergence point of the system is the 
	stable equilibrium of the differential equation counterpart.
	This equilibrium point corresponds to the actual values of plant and $\frac{1-\beta}{1-\beta + \alpha} \theta_R$.
	Moreover, this theorem proves that the estimated parameters
	converge with probability one to this equilibrium point which results in $\theta_M \rightarrow\frac{1-\beta}{1-\beta + \alpha} \theta_R$ with probability one.
%
\end{proof}

\begin{remark}
	Assumption \ref{ass:thm:gamma} can be satisfied by a broad range of gain sequences $\gamma(k)$. For instance,
	both regularity conditions hold for $\gamma(k)=\frac{C}{k^{\alpha}}$ when $0<C<\infty$ and $0<\alpha\le 1$.
\end{remark}
\begin{remark}
	Assumption \ref{ass:thm:A} requires monitoring the roots of $\hat{A}(\qi)$ polynomial. 
	This is a common issue in adaptive control and several methods have been proposed.
	For instance, the estimates can be projected to the interior of $\mathscr{D}_A$ whenever the poles fall out of (or on) the unit circle.
	Assumption \ref{ass:thm:B} requires monitoring the magnitude of polynomials $B(e^{-j\omega_h})$ at all compensation frequencies.
	The left inequality guarantees that $\hat{D}_B$ is always invertible. The right inequality requires some very rough knowledge about the plant magnitude because the term $\alpha/(1-\beta)$ is large according to \eqref{eq:alpha_beta_relation}. Both inequalities can be satisfied by projecting the estimates into the interior of $\mathscr{D}_B$ whenever they do not belong to $\mathscr{D}_B$.
\end{remark}

Assumption \ref{ass:thm:H} is 
in general difficult to verify since $B(\qi)$ polynomial is unknown. 
Based on theorem \ref{thm:main}, the 
equilibrium point of the closed loop system satisfies this assumption.
It can be shown that one of the factors	 that determines the domain of attraction associated with this equilibrium point is the excitation sequence $u(k)$ 
intensity. 
This implies that when no information about $B(\qi)$ is available and the estimated parameters are initialized with zeros, it may be required to use a large excitation signal such that the domain of attraction includes the initial values.
However, a nominal system is usually known in practice and this requirement can be relaxed. Moreover, when the system is slowly time-varying, it is expected that this assumption is satisfied with a significantly small excitation since the current estimates are kept inside the domain of attraction of the slowly varying equilibrium point.

\section{Empirical Study}\label{sec:dir_results}
This section provides the experimental verification of the proposed controller. 
The method is used to design a plug-in controller for tracking nano-scale unknown periodic trajectories with  high
frequency spectra in hard disk drives (HDDs).

A so-called \emph{single-stage} HDD uses a voice coil motor (VCM) for movements of the read/write heads \cite{horowitz2007dual}. 
The block diagram in Fig.~\ref{fig:SimpleFeedBackAndAdaptive_Abstract} can be adopted for this mechatronic system in track-following mode. 
The blocks $G$ and $C_F$ refer to the VCM and the nominal feedback controller respectively. 
The signals $w$, $r$, $n$ and $m$ denote the (unknown) airflow disturbance, \emph{repeatable runout} (RRO), \emph{non-repeatable runout} (NRRO) and {measurement noise} respectively. The measured \emph{position error signal} (PES) is denoted by $e$.
The design of $C_F$ is not discussed here, and it is assumed that this compensator can robustly stabilize the closed loop system and attenuate the broad band noises $w$, $n$ and $m$ (c.f. \cite{shahsavari2013robust,shahsavari2013h}).  
The plug--in controller that will be explained shortly targets the RRO ($r$) which consists of sinusoids.
An exact dynamics of the actuators is not known for each individual HDD. 
Furthermore, temperature variations and deterioration over time can introduce more uncertainties
\cite{malang2005method}.
We do not use any information about the system dynamics or the feedback controller. In other words, everything is unknown to us except the position error signal (PES).
Bit patterned media (BPM) is an emerging magnetic recording technology in which each data bit is recorded on a single  
magnetic island in a gigantic array patterned by lithography on the  disk.
%
This makes the servo
system a crucial component, 
and introduces significant new complexity. 
%
BPMR requires that the data tracks be followed with significantly more accuracy than what is required in conventional continuous media recording
since the head has to be accurately positioned over the  single--domain islands in order to read or write data. 
Unknown track shapes 
results in written--in runout which becomes repeatable (RRO) from the controller sight of view when the disk spins -- i.e. $r$ in Fig.~\ref{fig:SimpleFeedBackAndAdaptive_Abstract}. 
In our setup, RRO  has narrow--band contents at the HDD spinning frequency ($120Hz$) and its $173$ higher harmonics.
In other words, $n$ in \eqref{eq:r_definition} is $174$, 
\begin{align}\label{eq:omega_dual_stage}
\omega_i(rad/s)=i\times120(Hz)\times2\pi\qquad i=1,\cdots, 174,
\end{align}
and system sampling frequency is 41.760KHz ($120Hz\times 348$).

\subsection{Computer Simulation Results}
The magnitude response of the closed loop dynamics from the VCM input to the PES decays notably after $2KHz$, which makes the VCM at frequencies above $7KHz$  ineffective.
Accordingly, we only focus on tracking the first $58$ harmonics ($120Hz$, $240Hz$, ..., $6960Hz$) in the design of $C_A$ in Fig.~\ref{fig:SimpleFeedBackAndAdaptive_Abstract}.
The remaining 115 harmonics should be allocated to a higher bandwidth actuator which is beyond the scope of this paper and is considered as one of our future work.

The design parameters of the adaptive control algorithm are listed in Table~\ref{tab:parameters-simulation}.
\begin{table}[!t]
	\caption{Hyper Parameters of the adaptive control algorithm in the empirical study.}
	\label{tab:parameters-simulation}
	\centering
	\begin{tabular}{c c c }
		\hline\hline
		$n_A$ \eqref{eq:direct_theta_phi}	& $\alpha$ \eqref{eq:integral} 	&$\beta$ \eqref{eq:integral}		
	\\
	\hline
		$5$		&4\text{\sc{e}-}5		&1-(2\text{\sc{e}-}7)				
		\\
		\hline\hline
	\end{tabular}
\end{table}
%
%
%
The estimated coefficients for $A(q^{-1})$ and $B(q^{-1})$ that construct $\hat{\theta}_A$ and $\hat{\theta}_B$ are shown in Fig.~\ref{fig:AB-1-29}. The figure shows that the estimated parameters converge to ``some'' values quickly. In order to evaluate the convergence point, we generated the transfer function $\frac{\hat{B}(q^{-1})}{\hat{A}(q^{-1})}$ that corresponds to these values. The frequency response of this (5\textsuperscript{th} order) transfer function is compared to the actual transfer function of the VCM loop (a realistic 50\textsuperscript{th} order model) in Fig.~\ref{fig:bode-1-29}. The shaded strip indicates the compensation frequency interval where the adaptive controller was active. 

\begin{figure}[t]
	\centering
\includegraphics[width=0.8\linewidth]{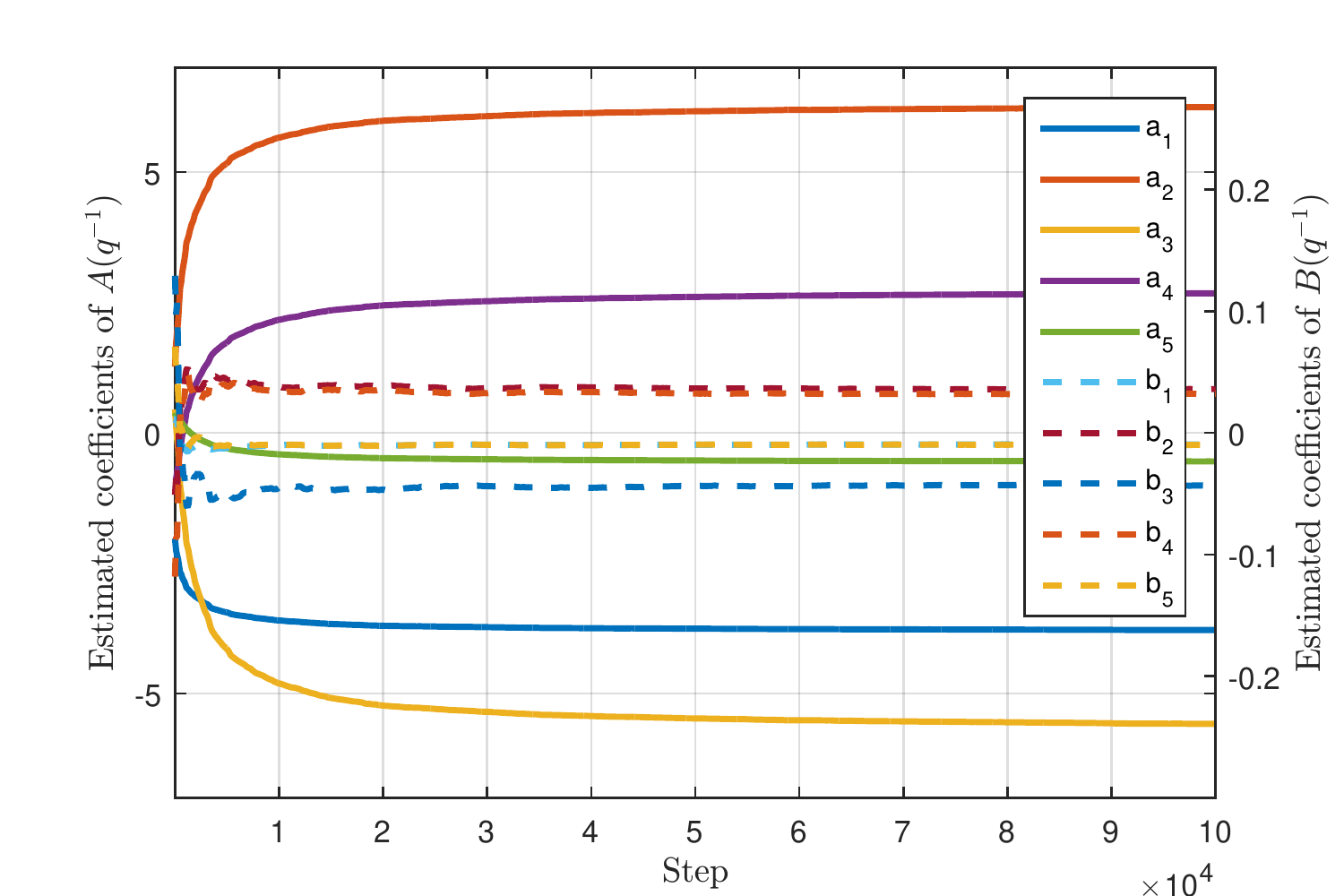}
	\caption{Evolution of $-\hat{\theta}_A$ and $\hat{\theta}_B$  elements. }
	\label{fig:AB-1-29}
\end{figure}
\begin{figure}[!t]
	\centering
	\includegraphics[width=0.9\linewidth]{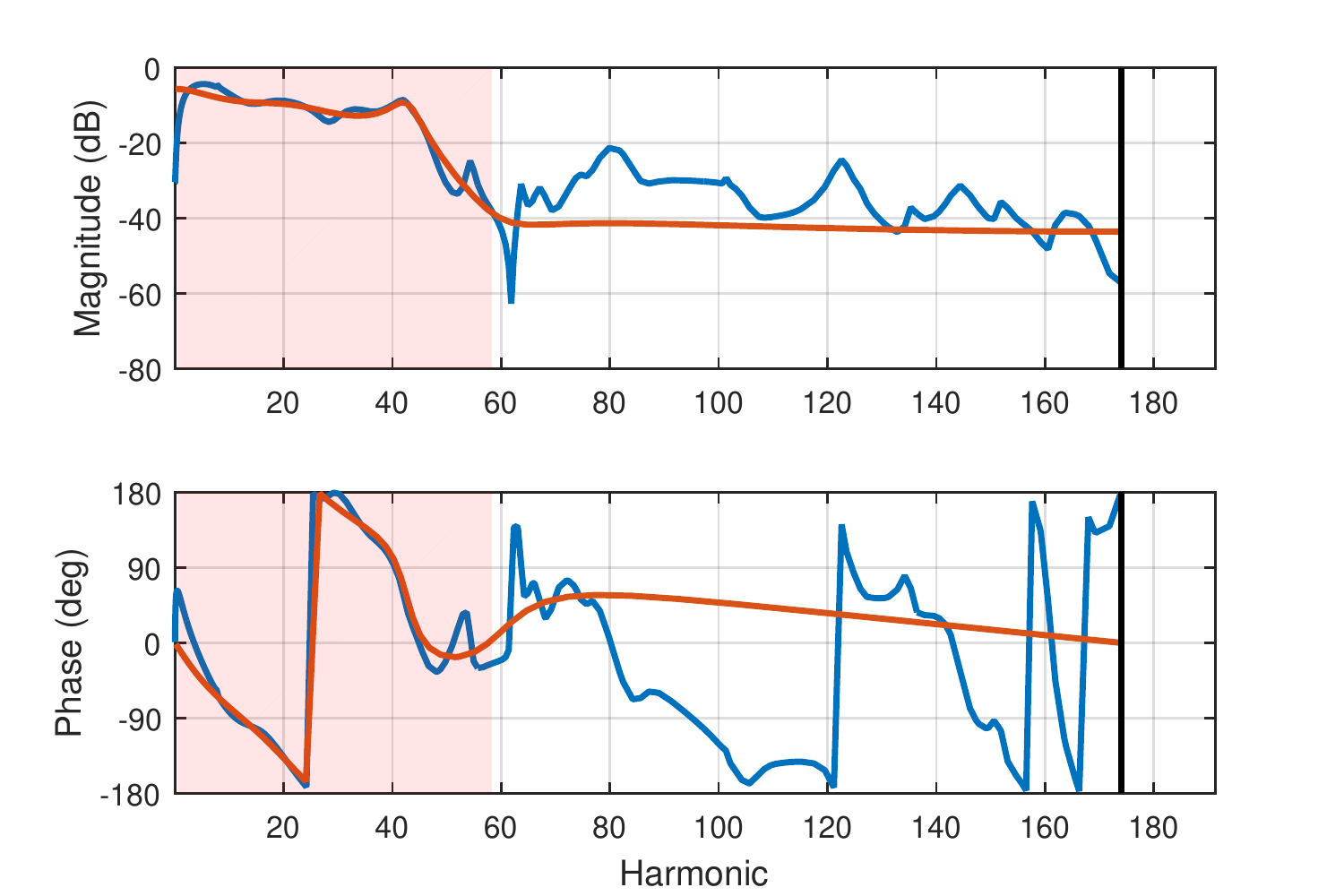}
	\caption{Frequency response comparison of the identified model and the actual VCM loop. The shaded strip indicates the compensation frequency range.}
	\label{fig:bode-1-29}
\end{figure}

The estimated residue parameters, $\hat{\theta}_M$, are depicted in Fig.~\ref{fig:thetaM-1-29}. 
The plot shows that the residual disturbance converges towards zero as the algorithm evolves. 
	\begin{figure}[!tb]
		\centering	\includegraphics[width=0.9\linewidth]{\figpath/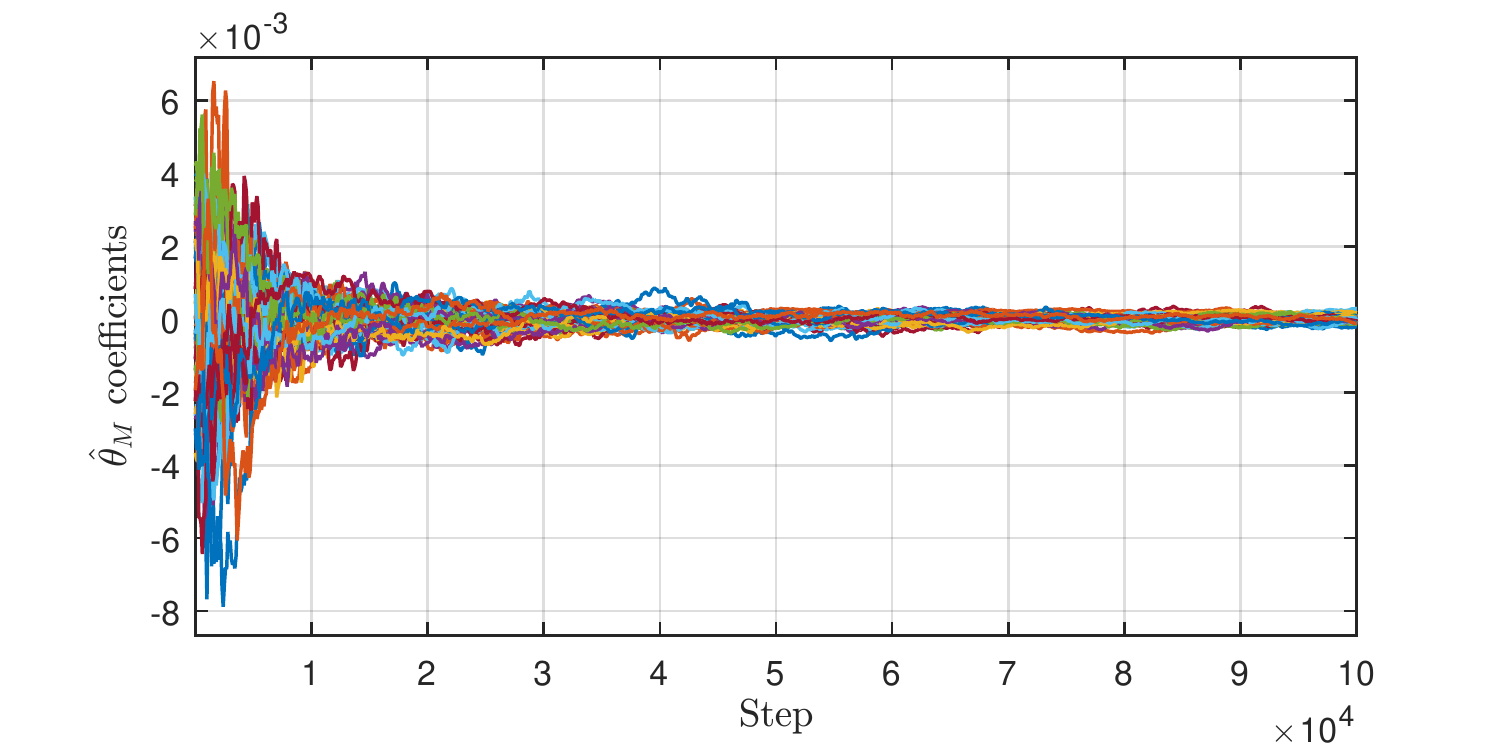}
		\caption{Estimated residue parameters, $\hat{\theta}_M$.}
		\label{fig:thetaM-1-29}
	\end{figure}
%
%
This can be verified in frequency domain based on the spectrum of error too.
The amplitude spectrum of the error before and after plugging the adaptive controller to the closed loop servo system are depicted in 
Fig.~\ref{fig:spectrum-bars-1-29}. 
For clearness, the figure only shows the amplitude of the error Fourier transformation at compensation frequencies. 
	\begin{figure}[!t]
		\centering	\includegraphics[width=0.9\linewidth]{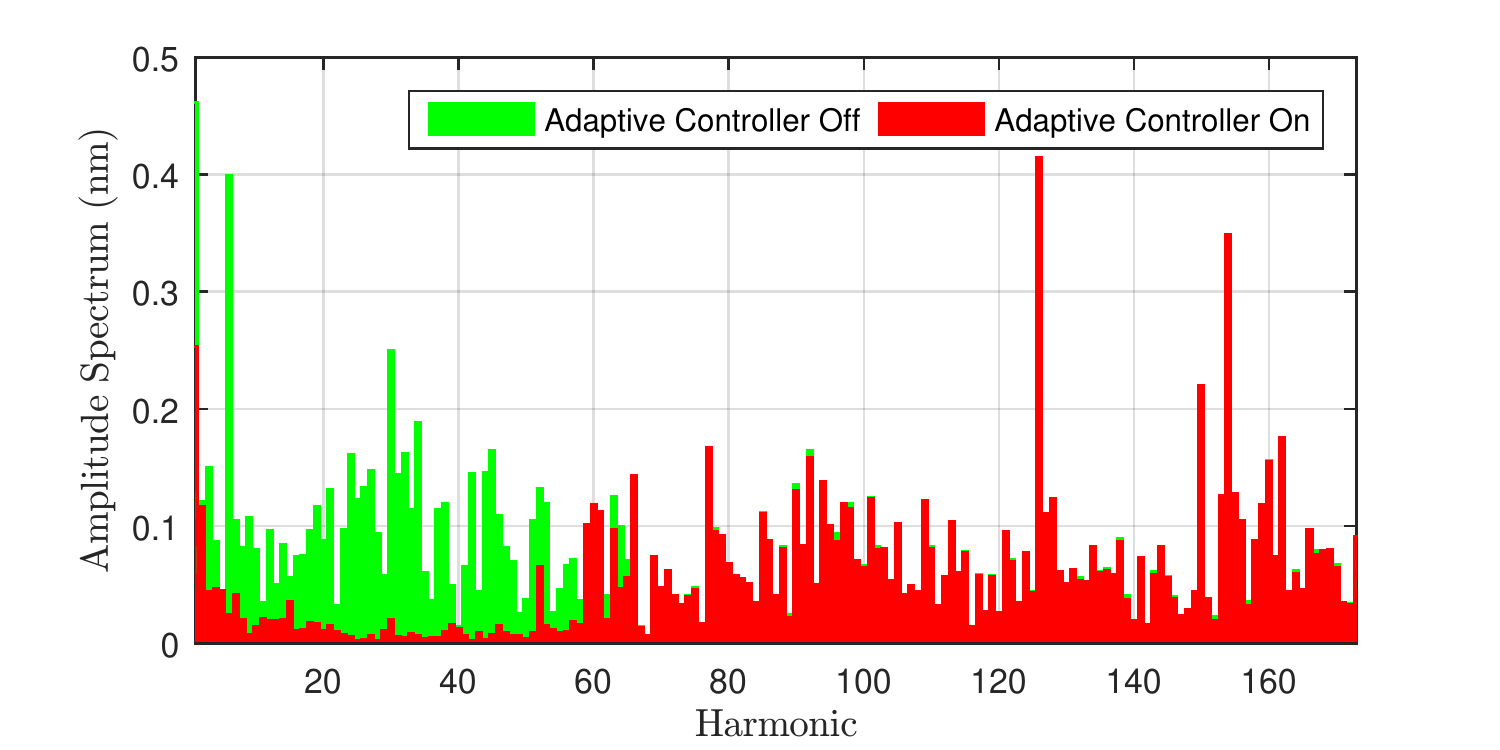}
		\caption{Comparison of the position error spectrum before and after plugging the adaptive controller. This figure shows the amplitude of Fourier transformation only at harmonics -- i.e. other frequencies are removed. }
		\label{fig:spectrum-bars-1-29}
	\end{figure}

As mentioned earlier, the control signal $u_A$ learned by the controller is periodic. Figure~\ref{fig:uffwd-1-29} depicts one period of this signal, which can be saved as one period of a repetitive feedforward control sequence that is able to perfectly compensate the first 58 harmonics.

	\begin{figure}[!t]
		\centering	\includegraphics[width=0.9\linewidth]{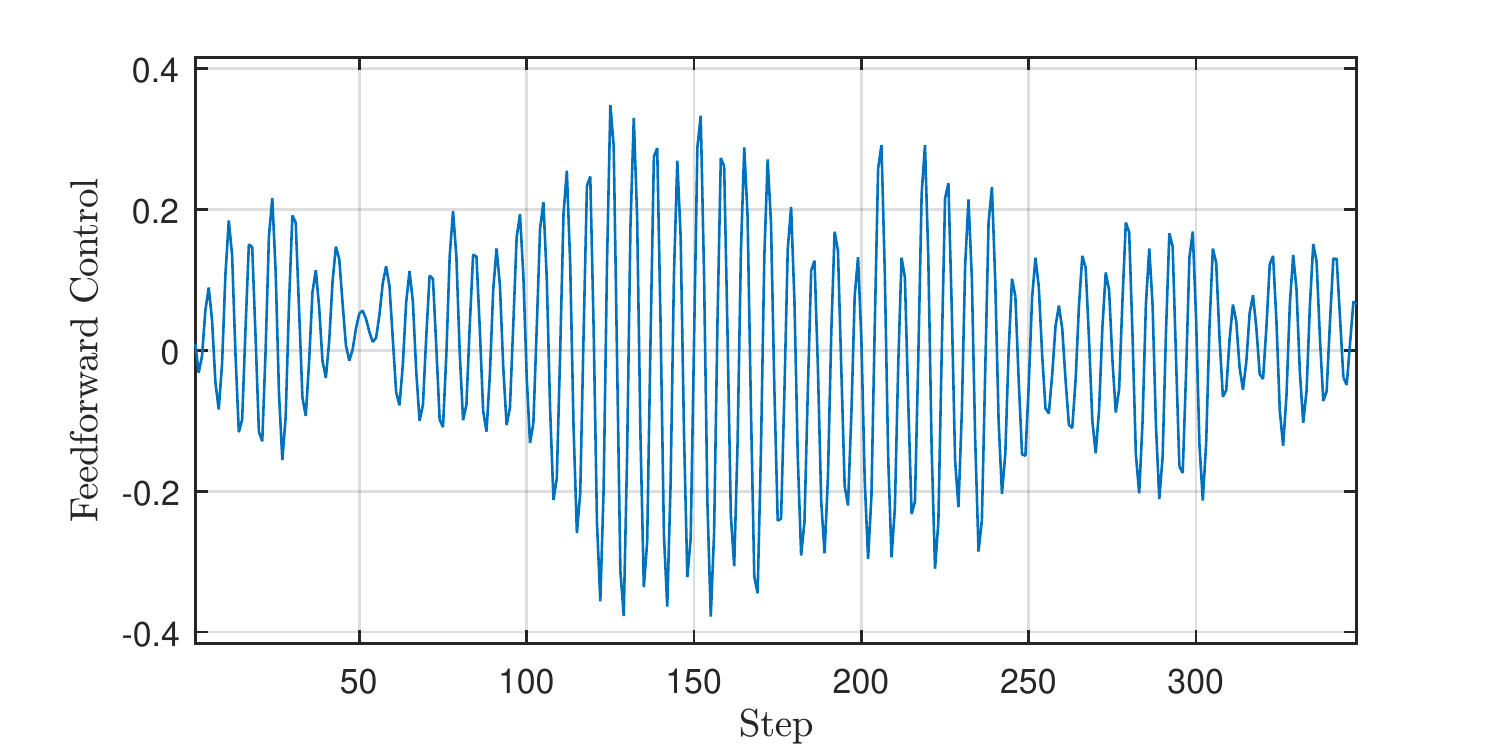}
		\caption{Feedforward control signal, $u_A(k)$, learned by the adaptive controller.}
		\label{fig:uffwd-1-29}
	\end{figure}

\section{Conclusion}\label{sec:conclusion}
A novel direct adaptive control method for the rejection of disturbances or tracking trajectories consisted of multiple sinusoidals with selective frequencies was proposed. The method is applicable to both minimum and non-minimum phase linear systems with unknown dynamics. 
The adapted parameters converge to the real values when a large enough excitation signal is injected to the system. In the presence of some rough knowledge about the system dynamics, the excitation signal can be reduced considerably. 
The analysis in this paper was performed for linear time-invariant systems. However, similar results can be extended to systems with slowly time--varying parameters. 

We verified the effectiveness of the proposed control algorithm in tracking unknown nano--scale  periodic trajectories in hard disk drives by designing an add--on repetitive controller that was able to track the first 58 harmonics of the disk spinning frequency. Full spectrum compensation was impossible in our running example due to the VCM limited bandwidth. This issue can be addressed by deploying a \emph{dual--stage} mechanism that has a high--bandwidth actuator in conjunction with the VCM. Extension of the proposed method to multi--input single--output systems and experimental verification of the algorithm will form our future work.


\bibliographystyle{ieeetr}        
\bibliography{../Thesis/thesis_phd/references,../Thesis/thesis_phd/references2}{}


\end{document}